\documentclass[12pt]{amsart}

\usepackage{amsmath}

\usepackage{amsfonts,amssymb,amsthm,enumitem}
\usepackage[margin=1.2in]{geometry}
\usepackage{mathtools}
\usepackage{nicefrac,setspace}
\usepackage{xcolor}

\usepackage{array}
\usepackage{hyperref}   
\hypersetup{
	colorlinks,
	linkcolor={black!80},
	citecolor={blue},
	urlcolor={blue}
}
\usepackage{url}            
\usepackage{booktabs}       
\usepackage{amsfonts}       
\usepackage{nicefrac}       
\usepackage{microtype}      

\makeatletter
\newcommand*\bigcdot{\mathpalette\bigcdot@{.5}}
\newcommand*\bigcdot@[2]{\mathbin{\vcenter{\hbox{\scalebox{#2}{$\m@th#1\odot$}}}}}
\makeatother

\newtheorem{theorem}{Theorem}
\newtheorem{claim}[theorem]{Claim}
\newtheorem*{theorem*}{Theorem}
\newtheorem*{remark}{Remark}
\newtheorem*{claim*}{Claim}
\newtheorem*{remark*}{Remark}
\newtheorem*{lemma*}{Lemma}

\newtheorem{lemma}[theorem]{Lemma}
\newtheorem{fact}[theorem]{Fact}

\newtheorem{prop}[theorem]{Proposition}

\renewcommand{\mod}{\mathop {\mathsf{mod}}}

\renewcommand{\P}{\mathop {\mathbb P}}

\newcommand{\eps}{\epsilon}

\newcommand{\R}{\mathbb{R}}	
\newcommand{\C}{\mathbb{C}}		 

\newcommand{\F}{\mathbb{F}}
\newcommand{\E}{\mathop \mathbb{E}}

\newcommand{\Expect}[1]{\mathop{\mathbb{E}}\left
	[ #1 \right ]}
\newcommand{\Ex}[2]{\mathop{\mathbb{E}}\displaylimits_{#1}\Big
	[ #2 \Big ]}
\def\rem#1{{\marginpar{\raggedright\scriptsize #1}}}

\title[The Fourier growth of bounded functions]{Tight bounds on 
the Fourier growth of bounded functions on the hypercube}

\begin{document}

\author[Iyer et al.]{Siddharth Iyer}
\address{School of Computer Science, University of Washington}
\email{sviyer97@gmail.com}
\author[]{Anup Rao}
\address{School of Computer Science, University of Washington}
\email{anuprao@cs.washington.edu}
\author[]{Victor Reis}
\address{School of Computer Science, University of Washington}
\email{voreis@cs.washington.edu}
\author[]{Thomas Rothvoss}
\address{School of Computer Science, University of Washington}
\email{rothvoss@uw.edu}
\author[]{Amir Yehudayoff}
\address{Department of Mathematics, Technion-IIT}
\email{amir.yehudayoff@gmail.com}

\begin{abstract}
We give tight bounds on the degree $\ell$  homogenous parts $f_\ell$ of a bounded function $f$ on the cube.  
We show that if $f: \{\pm 1\}^n \rightarrow [-1,1]$ has degree $d$, then $\| f_\ell \|_\infty$ is bounded by $d^\ell/\ell!$, and $\| \hat{f}_\ell \|_1$ is bounded by $d^\ell e^{\binom{\ell+1}{2}} n^{\frac{\ell-1}{2}}$.
We describe applications to pseudorandomness
and learning theory. We use similar methods to generalize the classical Pisier's inequality from convex analysis. Our analysis  involves properties
of real-rooted polynomials that may be useful elsewhere.
\end{abstract}

\maketitle

\section{Introduction}
The goal of complexity theory is to understand the space of functions that are efficiently computable. Every function  $f:\{\pm 1\}^n \rightarrow \R$ corresponds to a multilinear polynomial in $n$ variables, and under many models of computation, efficiently computable functions correspond to bounded  polynomials of low degree. This motivates an investigation of the characteristics of such functions.
Our main results are tight bounds on the magnitudes of coefficients.


A set $S \subseteq [n]$ corresponds to the \emph{monomial} or \emph{character}
$\chi_S(x) :=  \prod_{j \in S} x_j$.
Every $f: \{\pm 1\}^n \to \R$ can be uniquely expressed as 
$$f(x) = \sum_{S \subseteq [n]} \hat{f}(S) \cdot \chi_S(x),$$
where $\hat{f}(S) \in \R$ are the \emph{Fourier coefficients} of $f$.

\makeatletter
\newcommand{\thickhline}{%
	\noalign {\ifnum 0=`}\fi \hrule height 2pt
	\futurelet \reserved@a \@xhline
}
\newcolumntype{"}{@{\hskip\tabcolsep\vrule width 2pt\hskip\tabcolsep}}
\makeatother
\setlength{\arrayrulewidth}{.3mm}
\setlength{\tabcolsep}{18pt}
\renewcommand{\arraystretch}{1.3}
\begin{figure}[t]
	\begin{center}
		\begin{tabular}{p{7cm} " p{2.8cm} " p{1.3cm} } 
			Class of functions $f:\{\pm 1\}^n \rightarrow [-1,1]$ & $\|\hat{f}_\ell \|_1 \leq $ & Ref. \\ [0.5ex] 
			\thickhline
			CNFs of width $w$ & $w^{O(\ell)}$ & \cite{Mansour92} \\ 
			\hline 
			Width $w$ regular oblivious read-once branching programs& $(2w^2)^\ell$ & \cite{ReingoldSV13} \\ 
			\hline 
			Width $w$ oblivious read-once branching programs & $(O(\log n))^{w\ell}$ & \cite{ChattopadhyayHR18} \\ 
			\hline 
			Boolean functions of maximum sensitivity $s$ & $s^{O(\ell)}$ & \cite{GopalanSW16} \\	
			\hline
			$\F_2$ polynomials of degree $d$ & $\ell^\ell \cdot 2^{3d\ell}$ & \cite{CHH19} \\	
			\hline
			Decision trees of depth $d$ & $(O(\sqrt{d \log n}))^\ell$ & \cite{Tal20,SherstovSW20} \\	
			\hline
			Parity decision trees of depth $d$ & $d^{\ell/2} \cdot O(\ell \log n)^\ell$ & \cite{GirishTW21}
		\end{tabular}
	\end{center}
	\caption{Known bounds on the Fourier growth of various classes of functions} \label{fig1}
\end{figure}

Bounds on the Fourier coefficients play a key role in computer science
(see the textbook~\cite{o2014analysis} and also~\cite{mansour1994learning,BackursB14,FHKL2016,CHH19}
and references within). Typical results bound the growth of the $\ell_1$ norm of the Fourier coefficients in terms of their degree. 
The $\ell$-th homogenous part of $f$ is
$$f_\ell(x) :=\sum_{S \subseteq [n]: |S|=\ell} \hat{f}(S) \cdot \chi_S(x).$$ 
The main objective is to control the two norms $$ \|\hat{f}_\ell \|_1 : = \sum_{S \subseteq [n]: |S| = \ell} |\hat{f}(S) |,$$
and
$$ \|f_\ell \|_\infty := \max_{x \in \{\pm 1\}^n} |f_\ell(x)|.$$
By the triangle inequality, we must have $\|f_\ell \|_\infty \leq \|\hat{f}_\ell \|_1$. For general functions, the first quantity can be substantially smaller than the second. 
For symmetric functions $f$, the two quantities are the same: $\| \hat{f}_\ell \|_1 = |f_\ell(\mathrm{1}^n)| \leq \| f_\ell \|_\infty$.

Several works have proved non-trivial bounds on $\|\hat{f}_\ell \|_1$ for functions that are efficiently computable. Figure~\ref{fig1} lists some known results in this direction. 
An important motivation for bounding these norms is that a class of functions with small Fourier growth can be efficiently \emph{learned}~\cite{mansour1994learning}, and admits efficient \emph{pseudorandom generators}~\cite{ChattopadhyayGLLS20}.  A pseudorandom generator for a class of functions is a function that generates a distribution that uses a small random seed to generate a distribution that is supported on a small set, yet is indistinguishable from the uniform distribution to functions from the class. Chattopadhyay, Hatami, Hosseini and Lovett~\cite{CHH19} showed how to construct pseudorandom generators for any class of functions satisfying $\| \hat{f}_\ell \|_1 \leq t^\ell$, using $t^2 \cdot  \mathsf{polylog}(n)$ bits of seed. Similarly, Chattopadhyay, Gaitonde, Lee, Lovett and Shetty \cite{ChattopadhyayGLLS20} showed that bounds on $\|f_\ell \|_\infty$ also lead to efficient pseudorandom generators. Let $\mathcal{F}$ be a class of functions that is closed under restrictions
(i.e., setting a variable to $\pm 1$ keeps the function in $\mathcal{F}$).
Suppose there are parameters $k>2$ and $t>0$ such that every function $f \in \mathcal{F}$ satisfies $\|f_k \|_\infty \leq t^k$, then there is a pseudorandom generator of seed length $k \cdot  t^{2+4/(k-2)} \cdot  \mathsf{polylog}(n)$ for the class of functions $\mathcal{F}$.

Given these applications, it is interesting to ask for the most general bounds. What can we say about $\| \hat{f}_\ell \|_1$ and $\|f_\ell \|_\infty$ if $f: \{\pm 1 \}^n \rightarrow [-1,1]$ is an arbitrary function of degree $d$? Backurs and Bavarian \cite{BackursB14} and later Filmus, Hatami, Keller and Lifshitz \cite{FHKL2016} studied bounds on the \emph{influences} of such functions. With regards to the questions we study here, the techniques of \cite{FHKL2016} imply that $\| f_1\|_\infty = \|\hat{f}_1\|_1 \leq d$. In this work, we give tight bounds on $\|f_\ell \|_\infty$ and $\| \hat{f}_\ell \|_1$ for every $\ell$. 

Our methods are intimately connected to proofs of a classical result in convex analysis called \emph{Pisier's inequality}~\cite{pisier1979espaces,pisier1980theoreme}. Let $f:\{\pm 1\}^n \rightarrow \R^m$ be a vector valued  function. The $m$ coordinates of $f$ can be expressed as polynomials, so as before we have
$$f(x) = \sum_{S \subseteq [n]} \hat{f}(S) \cdot \chi_S(x),$$
where now $\hat{f}(S) \in \R^m$ is a vector. We define $f_\ell$ by projecting $f$ to its degree $\ell$ part,  just as we did earlier. 
 Pisier's inequality   says that every norm $\| \cdot \|$ on $\R^m$ must satisfy
$$ \Expect{\|f_1(X)\|^2}^{1/2} \leq O(\log(m+1)) \cdot \Expect{\|f(X)\|^2}^{1/2},$$ 
where $X \sim \{\pm1\}^n$ is uniformly distributed.


This inequality has important applications in geometry. 
Most strikingly, combined with a result of Figiel and Tomczak-Jaegermann~\cite{figiel1979projections} it implies the \emph{$MM^*$-estimate},
which says that in an average sense, symmetric convex bodies behave much more like ellipsoids than one could derive from John's theorem~\cite{john1948extremum}.
The $MM^*$-estimate is a central piece in the proofs of
Milman's QS-theorem~\cite{milman1985almost,milman1986inegalite,milman1988isomorphic} and the construction of \emph{$M$-ellipsoids}~\cite{milman1988isomorphic}, some of the most consequential results in convex geometry.

In our work, we generalize Pisier's inequality
to higher degrees, and make the proof more explicit
(see discussion in Section~\ref{sec:tech} below).

\subsection{Results}
Our results and proofs are intimately connected with the Chebsyshev polynomial $T_d(z)$. This is the unique polynomial of degree $d$ so that $T_d(\cos( \theta)) = \cos(d \theta)$. Denote by $C(d,\ell)$ the coefficient of $z^\ell$ in $T_d(z)$. Our first result is that these coefficients give upper bounds on $\|f_\ell\|_\infty$.

\begin{theorem} \label{thminftybound}
	If $f: \{\pm1\}^n \rightarrow [-1,1]$ has degree $d$, then
	\begin{align*}
	\| f_\ell\|_\infty \leq \begin{cases} |C(d,\ell)| & \text{if $d = \ell \mod 2$,}\\
	|C(d-1,\ell)| & \text{otherwise.}
	\end{cases}  
	\end{align*} 
\end{theorem}

To understand the theorem better, 
recall the known formula~\cite{cheb}:
\begin{align} \label{chebeq}
	C(d,\ell) = \begin{cases} (-1)^{(d-\ell)/2}\cdot 2^\ell \cdot \frac{d}{d+\ell}\cdot \binom{\frac{d+\ell}{2}}{\ell} & \text{if $d=\ell \mod 2$,}\\
	0 & \text{otherwise.}
	\end{cases}
\end{align}
We can use the arithmetic-mean-geometric-mean inequality and (\ref{chebeq}) to show 
\begin{align} \label{Cbound}
|C(d,\ell)| = 
2^\ell \cdot \frac{d}{d+\ell}\cdot 
\frac{1}{2^\ell \ell!}
\prod_{k = 0}^{\ell-1} (d+\ell-2k)
\leq
\frac{d^\ell} {\ell!}.
\end{align}
In particular, the theorem states that $\| f_\ell \|_\infty \leq \frac{d^\ell}{\ell!}$. The following proposition shows that the bound cannot be significantly improved when $n \gg d$:
\begin{prop}
\label{prop:const1}
For every $n,d,\ell$ such that $d= \ell \mod 2$, the bounded function $f(x) = T_d((x_1 + \dotsb + x_n)/n)$ satisfies $\| f_\ell\|_\infty \geq |C(d,\ell)| - 2 e^d (d+1)! / n$.
\end{prop}

We also provide general bounds on the larger $\|\hat{f_{\ell}}\|_1$:
\begin{theorem}\label{thml1bound}
If $f: \{\pm1\}^n \rightarrow [-1,1]$ has degree $d$, then for $\ell \geq 1$,
	$\|\hat{f_{\ell}}\|_1 \leq n^{\frac{\ell-1}{2}} \cdot d^\ell \cdot e^{\binom{\ell+1}{2}}.$
\end{theorem}
Once again, we give an example matching this bound when $d = \ell \ll n$:

\begin{prop}
\label{prop:const2}
There is a homogenous degree $d$ polynomial $f: \{\pm 1\}^n\rightarrow [-1,1]$ so that $$\|\hat{f_{d}}\|_1 = \frac{1}{2} \cdot \sqrt{\frac{1}{n} \cdot \binom{n}{d}}.$$
\end{prop}

Our methods allow to prove the following generalization of Pisier's inequality.

\begin{theorem}
\label{thm:P}
Let $\ell,m,n$ be positive integers and $\| \cdot \|$ be a norm on $\R^m$.
Let $X$ be uniformly distributed in $ \{\pm 1\}^n$.
Then for any function $f: \{\pm 1\}^n \rightarrow \R^m$, 
$$ \Expect{\|f_\ell(X)\|^2}^{1/2} \leq 
\Big(4+\frac{ 6 \log(m+1)}{\ell} \Big)^\ell \cdot \Expect{\|f(X)\|^2}^{1/2}.
$$ 
\end{theorem}

Bourgain showed that Pisier's inequality is sharp~\cite{bourgain1984martingales}.
An adaptation of his construction shows that Theorem ~\ref{thm:P} is also sharp, though we omit the details here.

We conclude this section with an application to  learning theory.
Suppose we want to approximate an unknown function $f$.
Access to $f$ is given by random queries of the form
$(X,f(X))$ where $X \sim \{ \pm 1\}^n$ is uniformly distributed.
The goal is to efficiently compute $g$ so that $\Expect{|f(X)-g(X)|^2} \leq \varepsilon$.
This problem was studied in several works
(see e.g.~\cite{mansour1994learning} and references within).
The theorems above lead to improving 
the sample complexity from polynomial in $n^d$ 
to $o(n^d)$, for $d$ fixed and $n \to \infty$.

\begin{theorem}
\label{thm:learn}
There is a constant $c>1$ so that the following holds.
  Let $f : \{ \pm 1\}^n \to [-1,1]$ be of degree $d \geq 1$ and let $\varepsilon>0$.
From $N \leq 2^{c d^2}\frac{n^{d-1}\log(n)}{\varepsilon^3}$ random queries to $f$, we can efficiently construct a function $g : \{ \pm 1\}^n \to \R$ with $\Expect{|f(x)-g(x)|^2} \leq \varepsilon$.
\end{theorem}

The function $g$ is obtained by estimating the large Fourier coefficients of $f$. The analysis
closely follows classical arguments that can be found for example in~\cite{mansour1994learning,o2014analysis}.


\subsection{Outline}

Theorem~\ref{thminftybound} is proved in Section~\ref{sec:boundedfunctions1}, Theorem~\ref{thml1bound} is proved in Section~\ref{sec:boundedfunctions2},
and Theorem~\ref{thm:P} is proved in Section~\ref{sec:Pisier}. 
Propositions~\ref{prop:const1} and~\ref{prop:const2} are proved in Section~\ref{sec:B}.
Theorem~\ref{thm:learn} is proved in Section~\ref{sec:learn}.

\subsection{Techniques}
\label{sec:tech}
Here we give a high-level sketch of some of our proofs, omitting many details that are explained later. 
Our techniques are inspired by proofs of Pisier's inequality. Pisier's original proof used complex analysis and interpolation. Bourgain and Milman found a different and more direct proof~\cite{bourgain1987new}.
Their proof relies on the Hahn-Banach theorem, the Riesz representation theorem, and Bernstein's theorem from approximation theory. These deep results are used to show that there is a function that is close to the linear function $L(x) = x_1+\dotsc+x_n$, yet has much smaller $\ell_1$ norm than $L(x)$. The existence of this linear proxy is proved in a clever but non-constructive way.

In our work, we give an explicit formula for a (more general) proxy with the properties alluded to above. Our key technical contribution is an explicit \emph{filter} that can be used to project a polynomial in $\cos(\theta)$ to its degree $\ell$ part.   
The filter is a central component of the proxy,
and can potentially be useful elsewhere.

\begin{theorem}
	\label{thm:filter}
	For every  $d \ge \ell$  with $d=\ell \mod 2$, there is a function $\phi : [0, 2\pi) \rightarrow  \R$ and a distribution on $\theta$ such that $\E[|\phi(\theta)|] =  |C(d,\ell)|$ and 
	\begin{align} \label{filtereq}
		\E[\phi(\theta) \cos^k (\theta)] = \begin{cases} 1 & \text{if $k = \ell$,}\\ 0 & \text{if $k\neq \ell$, $k \le d+1$.} \end{cases}
	\end{align}	 
\end{theorem}

The theorem cannot be improved, in the sense that 
any function $\phi$ satisfying  \eqref{filtereq} must also satisfy   $$\E[|\phi(\theta)|] \geq |\E[\phi(\theta) \cos(d \theta)]| = |C(d,\ell)|.$$ 
The proof of Theorem~\ref{thm:filter} is based on properties of Chebyshev polynomials and some non-trivial facts about real-rooted polynomials that may be of independent interest. If $p(z) = \sum_{j=0}^d c_j z^j$ is a polynomial, we write $p_{>k}(z) = \sum_{j=k+1}^d c_j z^j$. The proof of Theorem \ref{thm:filter} relies on the following theorem:
\begin{theorem}
	\label{thm:polyalternate}
	Let $p(z)$ be a real-rooted degree-$d$ polynomial whose roots are all positive. Then for every root $r$ of $p(z)$ and every $k \in \{0,\dotsc,d\}$, we have  
	$ (-1)^{d - k-1} \cdot p_{> k}(r) \geq 0.$
\end{theorem}
We prove Theorem~\ref{thm:filter} in Section~\ref{filtersec} and  Theorem~\ref{thm:polyalternate} in Section~\ref{sec:trunc}. 
Theorem~\ref{thminftybound} is proved using the filter as follows.
We construct a proxy $P: \{\pm 1\}^n \rightarrow \R$ using the formula\footnote{This formula is inspired by earlier proofs of Pisier's inequality.}:
\begin{align*}
P(x) = \Ex{\theta}{ \phi(\theta) \cdot \prod_{j=1}^n (1+ \cos(\theta) \cdot x_j)}.
\end{align*}
When $X \sim \{\pm 1\}^n$ is uniformly distributed, we can bound 
$$\E[|P(X)|] \leq \E[|\phi(\theta)|] \leq |C(d,\ell)|.$$ 
By construction, we have
\begin{align} \label{filtereq2}
\hat{P}(S)  = \begin{cases}
1 & \text{if $|S| =\ell$,} \\
0 &\text{if $|S| \neq \ell$, $ |S| \leq d$,}
\end{cases}
\end{align}
so we can compute $f_\ell$ via convolution as $f_\ell  = f*P$. The properties of $P$ imply that the convolution with it cannot be large at any point.

Theorem~\ref{thml1bound} is proved by induction. In the proof, we apply a random restriction to $f_\ell$. We set each variable of $f_\ell$ randomly with probability $\tfrac{1}{\ell}$ and leave it unset with probability $1-\tfrac{1}{\ell}$. This gives a degree $\ell$ function $g$. We use  Khintchine's inequality  to bound $\| \hat{f}_\ell \|_1$ in terms of $\| \hat{g}_{\ell-1} \|_1$. Since $g$ is bounded by $\| f_\ell \|_\infty$,  induction combined with Theorem~\ref{thminftybound} can be used to prove  Theorem~\ref{thml1bound}.

Theorem~\ref{thm:P} is proved by setting $d \approx \log(m+1)$ and using the proxy:
\begin{align*}
P(x) := 2^{\ell} \cdot \Ex{\theta}{ \phi(\theta) \cdot \prod_{j=1}^n (1+ \tfrac{\cos(\theta) \cdot x_j}{2})}
\end{align*}
Once again, the construction ensures that (\ref{filtereq2}) holds. 
Because $|\hat{P}(S)| \leq 2^{\ell-d} \cdot |C(d,\ell)|$ for $|S|>d$, 
the proxy $P$ is close to the symmetric homogenous polynomial of degree $\ell$ whose coefficients are all $1$. We can use the bound on $\E[|\phi(\theta)|]$ to bound $\E[|P(X)|] \leq 2^\ell |C(d,\ell)|$. 
Theorem~\ref{thm:P} is again proved by analyzing the convolution $f * P$.

\section{Preliminaries}\label{sec:prelim}

\subsection{Fourier analysis}

\begin{fact}[Parseval's identity]
\label{fact:parseval}
If $f : \{\pm 1\}^n \to \R^m$ and
$X \sim \{\pm 1\}^n$ is uniformly distributed
then 
$$ \Expect{ \|f(X)\|_2^2 }  = \sum_{S \subseteq [n]} \|\hat{f}(S)\|_2^2 .$$
\end{fact}

\begin{proof}
The proof is based on the orthonormality of the characters:
\begin{align*}
\Expect{ \|f(X)\|_2^2 }
& = \E \Big[ \Big\langle \sum_{S \subseteq [n]} \hat{f}(S) \cdot \chi_S(X), 
\sum_{T \subseteq [n]} \hat{f}(T) \cdot \chi_T(X) \Big\rangle \Big] \\
& = \sum_{S , T\subseteq [n]} \big\langle  \hat{f}(S) , 
 \hat{f}(T)  \big\rangle \E \Big[ \chi_S(X) \chi_T(X)  \Big] \\
& = \sum_{S \subseteq [n]} \| \hat{f}(S) \|_2^2. \qedhere
\end{align*}

\end{proof}


Convolution is a powerful tool when there is an underlying group structure.
Here the group is the cube $\{\pm 1\}^n$
with the operation $x \odot z = (x_1 z_1, \dotsc , x_n z_n)$. 
The convolution of a (vector-valued) 
function $f: \{\pm 1\}^n \rightarrow \R^m$ and 
a (scalar-valued) function $g:\{\pm 1\}^n \rightarrow \R$ is 
the function $f*g: \{\pm 1\}^n \to \R^m$ defined by
\begin{align*}
 f*g(x) = \Ex{Z}{g(Z) \cdot f(x \odot Z)}
\end{align*}
where $Z$ is uniformly random in $\{\pm 1\}^n$. 
We list some basic properties of convolution.

\begin{fact}
\label{fact:lin}
If $L:\R^m \to \R^m$ is a linear map then $L(f*g) = L(f) * g$. 
\end{fact}

\begin{fact}
\label{fact:f*gS}
One has $\widehat{f*g}(S) = \hat{g}(S) \cdot \hat{f}(S) $ for every $S \subseteq [n]$.
\end{fact}

\subsection{Norms and convexity}

\begin{fact}[Jensen's inequality] \label{fact:jensen}
	Given a convex function $f$ and a random variable $X$, we have $f(\mathbb{E}[X]) \le \mathbb{E}[f(X)]$.
\end{fact}

\begin{fact} \label{fact:l1} For any norm $\| \cdot \| : \R^m \to \R_{\geq 0}$ and functions $f : \{\pm 1\}^n \to \R^m$ and $g : \{ \pm 1\}^n \to \R$ one has 
	$$\Expect{\|f*g(X) \|^2}^{1/2} \leq \Expect{|g(X)|} \cdot \Expect{\|f(X) \|^2}^{1/2}.$$
where $X \sim \{ \pm 1\}^n$ uniformly.
\end{fact}
\begin{proof}
  We bound
	\begin{align*}
	\Expect{\|f*g(X) \|^2} & = \Ex{X}{\Big\|\Ex{Z} {g(Z) \cdot f(X \odot Z)} \Big\|^2}	\\ & \leq \Ex{X}{\Big(\Ex{Z}{|g(Z)| \cdot  \|f(X \odot Z)\|}\Big)^2},
	\end{align*}
	where the inequality follows from the convexity of the norm $\| \cdot \|$.  By the Cauchy-Schwarz inequality, we can continue
	\begin{align*}
	& \leq \Ex{X}{\Ex{Z}{|g(Z)|} \cdot \Ex{Z'}{|g(Z')|\cdot  \|f(X \odot Z')\|^2}} \\
	&= \Ex{Z}{|g(Z)|} \cdot \Ex{Z'}{|g(Z')| \cdot \Ex{X}  {\|f(X)\|^2}} \\
	&= \Big(\Ex{Z}{|g(Z)|}\Big)^2 \cdot \Ex{X}  {\|f(X)\|^2}. \qedhere
	\end{align*}
\end{proof}
It is convenient to replace a norm with the Euclidean norm. For this, we use the following standard result in convex geometry.
\begin{fact}[John's Theorem ~\cite{john1948extremum}]
\label{fact:john}
For any norm $\| \cdot \|$ on $\R^m$, there is an invertible linear map $J : \R^m \to \R^m$
such that for every $x \in \R^m$,
$$\|J(x)\|_2 \leq \|x\| \leq \sqrt{m} \cdot \|J(x)\|_2.$$
\end{fact}

\subsection{Some useful inequalities}

\begin{fact}[Stirling's approximation] \label{fact:stirling}
	For every $n \in \mathbb{N}$, we have $$ \sqrt{2 \pi} \cdot n^{n+1/2} e^{-n} \leq n! \leq e \cdot n^{n+1/2} e^{-n} .$$
\end{fact}

\begin{fact}[Khintchine's inequality]\label{fact:Khintchine}
	Let $Y \in \{\pm 1\}^n$ be uniformly random. For every integer $k>0$, there exist constants $A_k,B_k>0$ such that for every $x\in \R^n$,
	\begin{align*}
	A_k \|x\|_2 \leq \mathbb{E}\Big[\Big\lvert\sum_{i=1}^n  Y_ix_i\Big\rvert^{k} \Big]^{1/k} \leq B_k \|x\|_2.
	\end{align*}
	We can take $A_1 = \tfrac{1}{\sqrt{2}}$ and $B_k = k!$.
\end{fact}

\begin{lemma}[Chernoff bound] \label{lem:ChernovBound}
	Let $X_1,\ldots,X_n \in [-1,1]$ be independent random variables with $\E[X_i]=\mu$
	for all $i \in [n]$. For every $\lambda \geq 0$, 
	\[
	\Pr\Big[ \Big| \mu - \tfrac{1}{n}\sum_{i=1}^n X_i \Big| \geq \lambda \Big] \leq 2\exp(-\lambda^2n/2) .
	\]
	
\end{lemma}

\begin{fact}[Bernstein's inequality~\cite{bernstein}] 
	\label{fact:bern}
	Let $X_1, \dots, X_n$ be independent zero-mean random variables with $|X_i| \le M$ for all $i \in [n]$. For every $t > 0$,
	\[
	\Pr\Big[\Big|\sum_{i=1}^n X_i\Big| \ge t\Big] \le 2 \exp\Big(-\frac{t^2}{2\sum_{i=1}^n \E[X_i^2] + \frac{2}{3} Mt}\Big).
	\]
\end{fact}

\begin{fact}[Minkowski's inequality] \label{fact:MinkowskiIneq}
  Let $1 \leq p < \infty$, let $\| \cdot \| : \R^m \to \R_{\geq 0}$ be a norm and let $X,Y$ be jointly
  distributed random variables on $\R^m$ so that $\E[\|X\|^p],\E[\|Y\|^p] < \infty$.
  Then $$\E[\|X + Y\|^p]^{1/p} \leq
  \E[(\|X\| + \|Y\|)^p]^{1/p} \leq
   \E[\|X\|^p]^{1/p} + \E[\|Y\|^p]^{1/p}.$$
\end{fact}

\subsection{Real-rooted polynomials}
\label{sec:realrooted}

	A univariate polynomial $p(z)$ over $\R$ is \emph{real-rooted} if for all $w \in \C$, the equality $p(w)=0$ implies that $w \in \R$.
Newton's inequality implies that the coefficients of real-rooted polynomials are log-concave.
	A sequence $c_0, \dots, c_d$ is \emph{log-concave} if $c_j^2 \ge c_{j-1} \cdot c_{j+1}$ for  $j \in [d-1]$.

\begin{fact}[e.g.~\cite{unimodal}]
	\label{fact:logconcave}
	Let $p(z) = \sum_{j=0}^d c_j z^j$ be a real-rooted polynomial with real coefficients. Then the sequence $c_0,\dotsc,c_d$  is log-concave.
\end{fact}

An important consequence is that the magnitudes of the coefficients of such polynomials are \emph{unimodal}.
	A sequence $a_0,\ldots,a_d$ is unimodal if 
	there is an index $m$ such that 
	$$a_0 \leq a_1 \leq \dotsc \leq a_m \geq a_{m+1} \geq \dotsc \geq a_d.$$

\begin{fact}
	\label{fact:unimodality}
	If $c_0, \dots, c_d$ is a log-concave sequence of positive numbers, then
	it is unimodal. 
\end{fact}

\begin{proof}
	Since $c_j > 0$, it follows that $c_j/c_{j-1} \ge c_{j+1}/c_j$ for $j \in \{1, \dots, d-1\}$, that is, the sequence of consecutive ratios is non-increasing. Thus if $m$ is the largest index with $c_{m}/c_{m-1} \ge 1$, it follows that $c_0 \le c_1 \le \dots \le c_m \ge c_{m+1} \ge \dots \ge c_d$.
\end{proof}


%

\section{Bounds on $\| f_\ell \|_\infty$} 
\label{sec:boundedfunctions1}
 In this section, we prove Theorem~\ref{thminftybound}  assuming Theorem~\ref{thm:filter}. We start by constructing a proxy $P: \{\pm 1\}^n \rightarrow \R$ that filters out $f_\ell$ from $f$ by convolution. If $d = \ell \mod 2$, we use the parameters $(d,\ell)$ to obtain $\phi$ as in Theorem~\ref{thm:filter}. If $d \neq \ell \mod 2$, we use the parameters $(d-1,\ell)$ to obtain $\phi$. 
 The proxy is defined as
\begin{align*}
P(x) := \Ex{\theta}{\phi(\theta)\cdot \prod_{j=1}^{n}\big(1+x_j\cos\theta\big)}.
\end{align*}
Property~\eqref{filtereq} implies that for any $S\subseteq [n]$,
\begin{align*}
\hat{P}(S)  = \begin{cases}
1 & \text{if $|S| =\ell$,} \\
0 &\text{if $|S| \neq \ell$, $ |S| \leq d$.}
\end{cases}
\end{align*}
Since $f$ has degree $d$,  Fact~\ref{fact:f*gS} implies that $f_\ell = f*P$. 
Because $f$ is bounded, 
for every $x \in \{\pm 1\}^n$,
\begin{align*}
|f_\ell(x)|= \big |\Expect{P(Z) \cdot f(x\odot Z)} \big| \leq \Expect{|P(Z)|}  ,
\end{align*}
where $Z \sim \{\pm 1\}^n$ is uniform.
%
Finally, 
\begin{align*}
\Expect{|P(Z)|} 
&\leq \E \Big [|\phi(\theta) |\cdot \Big |\prod_{j=1}^{n}(1+Z_j\cos\theta)\Big| \Big ] & \text{triangle inequality}\\
&=\E \Big [|\phi(\theta)| \cdot \prod_{j=1}^{n}(1+Z_j\cos(\theta)) \Big] &\text{$(1+Z_j \cos(\theta)) \geq 0$}\\
&=\E [|\phi(\theta)|] &\text{$\E[Z_j]=0$}\\
&\leq \begin{cases}
|C(d,\ell)| & \text{when $d = \ell \mod 2$,}\\
 |C(d-1,\ell)| & \text{otherwise.}
\end{cases} & \text{Theorem~\ref{thm:filter}}
\end{align*}

\section{Bounds on $\|\hat{f}_\ell \|_1$ \label{sec:boundedfunctions2}}
Here we prove Theorem~\ref{thml1bound} assuming Theorem~\ref{thminftybound}. 
The proof is by induction on $\ell$. 
When $\ell=1$, we have $ \| \hat{f}_1\|_1 = \| f_1 \|_\infty$, because we can pick an input $x$ for which $f_1(x) = \sum_{i=1}^n \hat{f}(\{i\}) x_i= \| \hat{f}_1\|_1$. And  Theorem~\ref{thminftybound} implies $\|f_1\|_\infty \leq d$.

For the induction step, let $\ell > 1$.
We apply a random restriction to $f_\ell$, and use induction on the degree-$(\ell-1)$ homogenous part of the restricted function. Let $Q \subseteq [n]$ be a subset of the variables sampled by including each variable independently with probability $\tfrac{1}{\ell}$, and let $Z \sim \{\pm 1\}^n$ be uniformly random
and independent of $Q$.
The random restriction of $f_\ell$ is
$$ g(x) :=  \sum_{S \subseteq [n]} \hat{f}_\ell(S) \cdot  \chi_{S \setminus Q}(x) \cdot \chi_{S \cap Q}(Z).$$
The main idea is to relate $\| \hat{f}_\ell \|_1$
and $ \| \hat{g}_{\ell-1}\|_1$.
\begin{lemma}
\label{lem:fvsg}
$\| \hat{f}_\ell \|_1 \leq  e \cdot \sqrt{\frac{ 2n}{\ell}}\cdot \Expect{ \| \hat{g}_{\ell-1} \|_1}$.
\end{lemma}

Before proving the lemma, we use it to complete the proof.
Since $g/ \|f_\ell \|_\infty$ is bounded and of degree at most $\ell$, 
\begin{align*}
\| \hat{f}_\ell \|_1 &\leq  e \cdot  \sqrt{\frac{ 2n}{\ell}}\cdot \| f_\ell \|_\infty \cdot n^{\frac{\ell -2}{2}} \ell^{\ell-1} e^{\binom{\ell}{2}} &
\text{induction \& Lemma~\ref{lem:fvsg}}\\
&\leq   e \cdot \sqrt{\frac{ 2n}{\ell}}\cdot \frac{d^\ell}{\ell!} \cdot n^{\frac{\ell -2}{2}} \ell^{\ell-1}  e^{\binom{\ell}{2}} & \text{Theorem \ref{thminftybound}}\\
&\leq     e \cdot \sqrt{\frac{2 }{\ell}}\cdot \frac{d^\ell \cdot  e^{\ell}}{\sqrt{2 \pi \ell} \cdot  \ell^\ell} \cdot n^{\frac{\ell -1}{2}} \ell^{\ell-1}  e^{\binom{\ell}{2}} & \text{Fact~\ref{fact:stirling}}\\
&\leq   \frac{e  }{\sqrt{\pi} \cdot \ell^2}\cdot d^\ell \cdot n^{\frac{\ell -1}{2}}   e^{\binom{\ell}{2}+ \ell} & \\
&\leq   d^\ell  \cdot n^{\frac{\ell -1}{2}} \cdot   e^{\binom{\ell+1}{2}}. & 
\end{align*} 

\begin{proof}[Proof of Lemma~\ref{lem:fvsg}]
Start by fixing a set $U \subseteq [n]$ of size $\ell-1$.
Denote by $1_{Q \cap U =\emptyset}$ the indicator random variable for the event
that $Q \cap U = \emptyset$.
The corresponding coefficient in $g$ is
\begin{align*}
\hat{g}(U) = 1_{Q \cap U =\emptyset} \cdot   \sum_{j \in Q \setminus U} Z_j  \cdot \hat{f}_\ell(U \cup \{j\}) .
\end{align*}
We first fix $Q$ and take the expectation over $Z$.
Denote by $S$ the random variable that is zero
if $Q = \emptyset$ and is equal to $\frac{1}{\sqrt{|Q|}}$
when $Q$ is not empty.
For $j \not \in U$,
let $S_j$ be the random variable that is zero if $j \not \in Q$
and is $\frac{1}{\sqrt{|Q|}}$ when $j \in Q$.
For every $Q$, we can bound
\begin{align*}
\E_{Z}[|\hat{g}(U)|] &= 1_{Q \cap U =\emptyset} \cdot  \Ex{Z}{ \Big | \sum_{j \in Q \setminus U} Z_j \cdot \hat{f}_\ell(U \cup \{j\})\Big |} &\\
& \geq  1_{Q \cap U =\emptyset} \cdot \sqrt{\frac{1}{2}  \sum_{j \in Q\setminus U} \hat{f}_\ell(U \cup \{j\})^2} & \text{Fact~\ref{fact:Khintchine}}\\
& \geq \frac{1}{\sqrt{2}} \cdot 1_{Q \cap U =\emptyset} \cdot S\cdot  \sum_{j \in Q\setminus U} |\hat{f}_\ell(U \cup \{j\})| & \text{Cauchy-Schwarz} \\
& \geq \frac{1}{\sqrt{2}}\cdot 1_{Q \cap U =\emptyset}  \cdot \sum_{j \not \in U} S_j \cdot |\hat{f}_\ell(U \cup \{j\})| .
& \text{$S \geq S_j$}
\end{align*}
We now take the expectation over $Q$ as well:
\begin{align*}
\Ex{Z,Q}{|\hat{g}(U)|} 
& \geq \frac{1}{\sqrt{2}}
\sum_{j \not \in U} |\hat{f}_\ell(U \cup \{j\})| \cdot 
\E_{Q}[1_{Q \cap U =\emptyset} \cdot  S_j] . 
\end{align*}
Because $\xi \mapsto \tfrac{1}{\sqrt{\xi}}$ is convex, for each $j \not \in U$ we use 
Jensen's inequality to bound 
\begin{align*}
\E_{Q}[1_{Q \cap U =\emptyset}\cdot  S_j]
& =\P[Q \cap U =\emptyset] \cdot \P[j \in Q] \cdot  
\E_{Q|Q \cap U =\emptyset , j \in Q} \Big[ \frac{1}{\sqrt{|Q|} }\Big] \\
& \geq \P[Q \cap U =\emptyset] \cdot \P[j \in Q] \cdot  
\frac{1}{\sqrt{\E_{Q|Q \cap U =\emptyset , j \in Q}[|Q|] }} .
\end{align*}
We have $\P[j \in Q] = \frac{1}{\ell}$,
and $\P[Q \cap U =\emptyset] = \big(1 - \tfrac{1}{\ell} \big)^{\ell-1} \geq 1/e$,
since $(1- 1/\ell)^{\ell-1}$ is decreasing in $\ell$ and converges to $1/e$.  We can compute:
$$\E_{Q|Q \cap U =\emptyset, j \in Q}[|Q|] 
= 1 +(n-(\ell-1)-1) \frac{1}{\ell} = \frac{n}{\ell}.$$
So, we can bound
\begin{align*}
\E_{Q}[1_{Q \cap U =\emptyset} \cdot S_j]
 \geq \frac{1}{e \ell} \cdot  
\sqrt{\frac{\ell}{ n} } = 
\frac{1}{e\sqrt{n \ell}} .
\end{align*}
Overall, for every $U$ of size $\ell-1$,
\begin{align*}
\Ex{Q,Z}{|\hat{g}(U)|} 
& \geq  
\frac{1}{e\sqrt{ 2n \ell }} \cdot 
\sum_{j \not \in U} |\hat{f}_\ell(U \cup \{j\})| . 
\end{align*}
Summing over $U$,
\begin{align*}
 \Ex{Q,Z}{ \|\hat{g}_{\ell-1} \|_1} 
& \geq  
\frac{1}{e\sqrt{2n \ell}}\cdot  \sum_{U: |U|=\ell-1} 
\sum_{j \not \in U} |\hat{f}_\ell(U \cup \{j\})| \\
& = \frac{1}{e\sqrt{2n \ell}}\cdot \sum_{S : |S|=\ell} \ell \cdot 
|\hat{f}_\ell(S)| \\
& = \frac{1}{ e} \cdot  
\sqrt{\frac{\ell}{2n}} \cdot \|\hat{f}_\ell\|_1 . \qedhere
\end{align*}

\end{proof}

\section{A Higher level Pisier inequality\label{sec:Pisier}}

In this section, we prove Theorem~\ref{thm:P} assuming Theorem~\ref{thm:filter}.
We can express $f_\ell$ as the convolution
of $f$ with the level function 
\[
L_\ell(x) := \sum_{S \subseteq [n]: |S| = \ell} \chi_S (x);
\]
see Fact~\ref{fact:f*gS}. In order to analyze the norm of $f_\ell$,
we construct a proxy $P$ that is close to $L_\ell$. 
Let $d$ be a parameter with $d = \ell \mod 2$ to be determined. 
Let $\phi$ be as in Theorem~\ref{thm:filter}. Define 
\begin{align*}
 P (x) &:= 2^\ell \cdot \Ex{\theta}{\phi (\theta) \cdot \prod_{j=1}^n \Big (1+ \frac{\cos(\theta) \cdot x_j }{2} \Big)}\\
 & =  \sum_{S \subseteq [n]}  2^\ell \Ex{\theta}{\phi(\theta) \cdot
 	\frac{\cos^{|S|}(\theta)}{2^{|S|}} } \cdot \chi_S(x).
 \end{align*}
We think of $P$ as a ``good'' proxy for $L_\ell$:
\begin{align*}
\Expect{\|f_\ell (X)\|^2}^{1/2}  
& = \Expect{\|f*L_\ell(X)\|^2}^{1/2} \\
&= \Expect{\|f*P(X) + f*(L_\ell-P)(X)\|^2}^{1/2}\\
& \leq \Expect{\|f*P(X) \|^2}^{1/2} + \Expect{\| f*(L_\ell-P)(X)\|^2}^{1/2}.
& \text{Fact~\ref{fact:MinkowskiIneq}}
\end{align*}
Next, we bound each of the two terms separately.

%
%
%

To bound the first term, apply Fact \ref{fact:l1},
\begin{align*}
\Expect{\|f*P(X) \|^2}^{1/2} & \leq \Expect{|P(X)|} \cdot \Expect{\|f(X) \|^2}^{1/2} .
\end{align*}
Similarly to the end of Section~\ref{sec:boundedfunctions1},
we may bound the $\ell_1$ norm of $P$ by
\begin{align*}
\Expect{|P(X)|} 
 &\leq 2^\ell \cdot \Expect{ | \phi(\theta)| \cdot \prod_{j=1}^n \Big (1+ \frac{\cos(\theta) \cdot X_j }{2} \Big)} \\ 
& = 2^\ell \cdot \Expect{ | \phi (\theta)|} \leq 2^\ell \cdot  |C(d,\ell)| .
\end{align*}

To bound the second term, use John's theorem (Fact ~\ref{fact:john}). 
There is an invertible linear map $J : \R^m \to \R^m$
so that for every $x \in \R^m$,
$$\|J(x)\|_2 \leq \|x\| \leq \sqrt{m} \cdot \|J(x)\|_2.$$
Using $J$ we can switch between $\|\cdot\|$
and $\|\cdot \|_2$:
\begin{align*}
\Expect{\| f*(L_\ell-P)(X)\|^2}^{1/2} 
&  \leq \sqrt{m} \cdot \Expect{\| J(f*(L_\ell-P)(X)) \|_2^2}^{1/2}  \\
&  = \sqrt{m} \cdot \Expect{\| J(f)*(L_\ell-P)(X)\|_2^2}^{1/2} 
& \text{Fact~\ref{fact:lin}} \\
& =  \sqrt{m} \cdot \sqrt{\sum_S \|\widehat{J(f)}(S)\|_2^2 \cdot (\widehat{L_\ell-P}(S))^2} . & \text{Facts~\ref{fact:parseval} and~\ref{fact:f*gS}}
\end{align*}
We now claim that 
\begin{align*}
\hat{P}(S)  & = \begin{cases}
1 & \text{if $|S| = \ell$,}\\
0 & \text{if $|S| \neq \ell, |S| \leq d$,}\\
\leq 2^{\ell - d} \cdot |C(d,\ell)| & \text{otherwise.}
\end{cases}
\end{align*}
The first two cases follow directly from the properties of $\phi$ and the formula for $P$. The last case follows from
\begin{align*}
|\hat{P}(S)|  \leq \frac{2^\ell}{2^{|S|}} \cdot \Expect{|\phi (\theta)|} \leq 2^{\ell-d} \cdot |C(d,\ell)| .
\end{align*}
%
%
%
%
%
%
%
We can continue to bound
\begin{align*}
\Expect{\| f*(L_\ell-P)(X)\|^2}^{1/2} 
& \leq \sqrt{m}\cdot  2^{\ell-d} \cdot  |C(d,\ell)|  \cdot \sqrt{\sum_S \|\widehat{J(f)}(S)\|_2^2} \\
&= \sqrt{m}\cdot  2^{\ell-d} \cdot  |C(d,\ell)|  \cdot \Expect{\| J(f(X))\|_2^2}^{1/2} & \text{Fact~\ref{fact:parseval}}\\
& \leq \sqrt{m}\cdot  2^{\ell-d} \cdot  |C(d,\ell)|  \cdot \Expect{\| f(X)\|^2}^{1/2}.
\end{align*}

Putting the two parts together,
\begin{align*}
\Expect{\|f_\ell(X)\|^2}^{1/2}  
& \leq 2^\ell \cdot |C(d,\ell)|  \cdot \Big(1+ \frac{\sqrt{m}}{2^d} \Big) \Expect{\|f(X) \|^2}^{1/2}.
\end{align*}
For $\ell < \tfrac{1}{2} \log(m+1)$, we can set $d$ to be the smallest integer that
is larger than $\tfrac{1}{2} \log(m+1)$ and has the same parity as $\ell$, so that 
\begin{align*}
2^\ell \cdot |C(d,\ell)| \cdot \Big(1+ \frac{\sqrt{m}}{2^d} \Big) 
& \le 2^\ell \cdot \frac{d^\ell}{\ell!} \cdot 2 
& \text{Fact ~\ref{fact:stirling}} \\
& \le \Big(\frac{6 \log (m+1)}{\ell}\Big)^\ell .
\end{align*}
For $\ell \ge \tfrac{1}{2} \log(m+1)$, we can set $d := \ell$ so that
\[
2^\ell \cdot |C(d,\ell)|  \cdot \Big(1+ \frac{\sqrt{m}}{2^d} \Big) = 2^{2\ell - 1} \cdot \Big(1+ \frac{\sqrt{m}}{2^\ell} \Big)  \le 4^\ell.
\]
\begin{remark}
There is a slightly more general version of Theorem~\ref{thm:P}.
The Banach-Mazur distance of the norm $\| \cdot \|$
from the Euclidean norm $\|\cdot\|_2$ is
$$D = \inf \{d \in \R: \exists T \in \mathsf{GL}_m \ \forall x \in \R^m \ \|T(x)\|_2 \leq \|x\| \leq d\cdot \|T(x)\|_2\},$$
where $\mathsf{GL}_m$ is the group of invertible linear transformations from $\R^m$ to itself. 
John's theorem states that always $D \leq \sqrt{m}$.
The above argument proves that, more generally, we can replace
the $C \log(m+1)$ term by $C \log(D+1)$.
\end{remark}

\section{Constructing the filter}
\label{filtersec}
Here we construct the filter $\phi$ and prove  Theorem~\ref{thm:filter}.
Let $\theta$ be uniformly distributed over the $2d$ equally spaced angles
$$\mathcal{D} = \Big\{0,\frac{\pi}{d}, \dots, \frac{(2d-1)\pi}{d}\Big\}.$$
An important property of this distribution is that for integer $a$, we have
\begin{align} \label{geom}
\E[e^{ia \theta}] = \begin{cases} 1 & \text{if}\ a = 0 \mod 2d \\ 0 & \text{otherwise.} \end{cases}
\end{align}
Define
\[
Q(z) := \prod_{j=0}^d \Big(z - \cos\Big(\frac{j\pi }{d}\Big)\Big) = \sum_{j=0}^{d+1} c_j z^j,
\]
for some $c_{j} \in \mathbb{R}$. Let $Q_{> \ell}$ denote the suffix of $Q$:
\[
Q_{> \ell} (z) := \sum_{j = \ell + 1}^{d+1} c_j z^{j}.
\]
The rational function $\frac{Q_{> \ell}(z)}{z^{\ell+1}}$ is 
a polynomial.
Finally, define
\[
\phi(\theta) := 2^{d-1} \cdot \cos(d \theta) \cdot \frac{ Q_{> \ell} (\cos (\theta))}{\cos^{\ell+1}(\theta)} .
\]
It remains to prove that the filter $\phi$ has the desired properties.

The following claim helps to understand the correlation 
of $\phi$ with powers of $\cos$.
\begin{claim} 
\label{claim:cos}
For integers $k, d \geq 0$, we have 
\[\E[\cos(d\theta) \cos^k(\theta)] = \begin{cases} 0 & \text{if}\  k \neq d \mod 2 \\ 0 & \text{if}\ k < d \\ 2^{-(d-1)} & \text{if}\ k = d. \end{cases}\]
\end{claim}

\begin{proof}
If $k \neq d \mod 2$, the symmetry $\cos(d\theta) = (-1)^d \cos(d(\pi + \theta))$ and the symmetry of the distribution of $\theta$ complete the proof. 

For $k < d$, we use the identity $\cos(\theta)  = \frac{e^{i\theta} + e^{-i\theta}}{2}$. Property~\eqref{geom} implies
\[
\E[\cos(d\theta) \cos^k(\theta)] = \E\Big[\frac{e^{id\theta} + e^{-id\theta}}{2} \cdot \Big(\frac{e^{i\theta} + e^{-i \theta}}{2} \Big)^k \Big] = 0.
\]

For $k = d$, the expectation reduces via \eqref{geom} to
\[
\displaystyle \E\Big[\frac{e^{id\theta} + e^{-id\theta}}{2} \cdot \Big(\frac{e^{i\theta} + e^{-i \theta}}{2} \Big)^d \Big] = \E\Big[\frac{e^{id\theta} + e^{-id\theta}}{2} \cdot \Big(\frac{e^{id\theta} + e^{-id \theta}}{2^d} \Big) \Big] = 2^{-(d-1)}. 
\]
\end{proof}

Now, we can prove \eqref{filtereq}. 
The argument is based on Claim ~\ref{claim:cos}.
We use the following terminology.
The expressions we consider are sums of terms of the form
$\cos(d \theta) \cos^k(\theta)$.
The {\em degree} of such a term is $k$.

For $k \le \ell - 1$, all terms in $\phi(\theta) \cos^k(\theta)$ have degree at most $d-1$. Claim~\ref{claim:cos} implies that $\E[\phi (\theta) \cos^k(\theta)] = 0$. 

For $k = \ell$, we have a single term of degree $d$, so that
\[
\E[\phi (\theta) \cos^\ell(\theta)] = \E[2^{d-1} \cos(d \theta) \cos^d(\theta)] = 1.
\] 
For $\ell+1 \le k \le d$, 
\begin{align*}
\E[\phi (\theta) \cos^k(\theta)] &= \E[2^{d-1} \cos(d \theta) \cdot Q_{> \ell} (\cos (\theta)) \cos^{k-(\ell+1)}(\theta)] \\& = \E[2^{d-1} \cos(d \theta) \cdot \underbrace{Q(\cos (\theta))}_{=0} \cos^{k-(\ell+1)}(\theta)] = 0 ;
\end{align*}
the second equality holds because we added terms in $Q$ of degree at most $\ell$, and $\ell + k-(\ell+1) \le d-1$. 

Finally, for $k=d+1$, we need one more observation.
Since $\cos(\theta) = -\cos(\pi - \theta)$, the distinct roots of 
the real-rooted polynomial $Q$ come in pairs of the form $r, -r$. 
So, there is a polynomial $q$ so that
\[
Q(z) = \begin{cases} z\cdot q(z^2) & \text{if $d = 0 \mod 2$,}  \\  q(z^2)  & \text{otherwise,} \end{cases}
\]
where 
\[
q(z) = \prod_{j=0}^{\lfloor(d-1)/2\rfloor} \Big(z - \cos^2 \Big(\frac{j\pi }{d}\Big)\Big) .
\]
Because $\ell = d \mod 2$, the coefficient $c_\ell$ in $Q$
is zero.
Similarly to the previous case, we can bound
\begin{align*}
\E[\phi (\theta) \cos^k(\theta)] &= \E[2^{d-1} \cos(d \theta) \cdot Q_{> \ell} (\cos (\theta)) \cos^{k-(\ell+1)}(\theta)] \\& = \E[2^{d-1} \cos(d \theta) \cdot \underbrace{Q(\cos (\theta))}_{=0} \cos^{k-(\ell+1)}(\theta)] = 0 ;
\end{align*}
here we additionally used that $c_\ell=0$.

Next, we turn to computing $\E[|\phi(\theta)|]$. The key claim is the following:
\begin{claim}
\label{claim:qsign}
 For all $\theta \in \mathcal{D}$, the sign of 
 $\frac{Q_{>\ell}(\cos(\theta))}{\cos^{\ell+1}(\theta)}$ is the same. 
\end{claim} 

\begin{proof}
The polynomial $q$ has positive roots corresponding to 
nonzero $\cos^2(\theta)$ for $\theta \in \mathcal{D}$.
Because $d = \ell \mod 2$, 
the sign of $\frac{Q_{>\ell}(\cos(\theta))}{\cos^{\ell+1}(\theta)}$ is the same as the sign of 
$q_{> k} (\cos^2(\theta))$ for $k = \lfloor \tfrac{\ell-1}{2} \rfloor$.
Theorem~\ref{thm:polyalternate} completes the proof.
\end{proof}

Claim~\ref{claim:qsign} implies that the sign of $\phi(\theta)$ is determined by the sign of $\cos(d\theta)$. 
We can finally compute
\begin{align*}
\E[|\phi(\theta)|] 
& = |\E[\phi(\theta)  \cos(d \theta)]| & 
\text{$|\cos(d\theta)| = 1$ for $\theta \in \mathcal{D}$} \\
& =  |C(d,\ell)|. & \text{property~\eqref{filtereq}
\& definition of $T_d$}
\end{align*}

\section{On Real Rooted Polynomials\label{sec:trunc}}
In this section, we prove Theorem~\ref{thm:polyalternate}. First, we need a useful property of unimodal sequences.
\begin{claim}
	\label{clm:unimoal+zero}
	Let $a_0,\ldots,a_d$ be a unimodal sequence of positive numbers
	so that $\sum_{j=0}^d (-1)^j a_j = 0$.
	Then for all $k \in \{0,\ldots,d\}$,
	we have $(-1)^k \sum_{j=0}^k (-1)^j a_j \geq 0$.
\end{claim}

\begin{proof}
	Let $m$ be the position of a maximum of the unimodal sequence.
	For $k \leq m$ even,
	\begin{align*}
	\sum_{j=0}^k (-1)^j a_j
	& = a_0 + \sum_{j=1}^{\frac{k}{2}} (a_{2j} - a_{2j-1}) \geq 0.
	\end{align*}
	For $k \leq m$ odd,
	\begin{align*}
	\sum_{j=0}^k (-1)^j a_j
	& = \sum_{j=0}^{\frac{k-1}{2}} (a_{2j+1} - a_{2j}) \leq 0.
	\end{align*}
	This proves the claim when $k \leq m$. 
	A symmetric argument can be applied to 
	the suffix sums to conclude that for $m \leq k <d$,
	$$(-1)^{d-k+1} \sum_{j=k+1}^d (-1)^{d-j} a_j \geq 0.$$
	Together with the condition
	$\sum_{j=0}^d (-1)^j a_j = 0$, this implies that when $k > m$, 
	\begin{align*}
	(-1)^k \cdot \sum_{j=0}^k (-1)^j a_j
	& = (-1)^{k+1} \cdot \sum_{j=k+1}^d (-1)^j a_j\\
	& = (-1)^{d-k+1} \sum_{j=k+1}^d (-1)^{d-j} a_j \geq 0.
	\end{align*}
	
\end{proof}

Now we turn to proving Theorem ~\ref{thm:polyalternate}. 
\begin{proof}[Proof of Theorem~\ref{thm:polyalternate}]
	Write $p$ as $$p(z) = \prod_{j=1}^d (z - r_j) = \sum_{j=0}^d c_j z^j,$$ with $r_1,\dotsc, r_d>0$. For every $j \in \{0,\ldots,d\}$,
	we have $(-1)^{d+j} \cdot c_j > 0$. 
 	So, by Fact \ref{fact:logconcave}, the sequence $|c_0|,\dotsc, |c_d|$ is log-concave.
	
	Now, let $r$ by any root of $p$, and set
	$a_j = |c_j| r^j$. 
	Because the product of log-concave sequences is log-concave, the sequence $a_0,\dotsc, a_d$ is log-concave and positive. 

	By Fact~\ref{fact:unimodality}, 
	the sequence $a_0,\dotsc, a_d$ is unimodal.
	Because $r$ is a root of $p$, we know
	$p(r) = \sum_{j=0}^d (-1)^j a_j = 0$.
	And Claim~\ref{clm:unimoal+zero} implies
	\begin{align*}
	(-1)^{d - k-1} \cdot p_{> k}(r) & = (-1)^{d-k-1} \cdot \sum_{j = k+1}^d c_j r^j\\
	& = (-1)^{-k-1} \cdot  \sum_{j  = k+1}^d (-1)^{j} a_j\\
	& = (-1)^{k} \cdot \sum_{j  = 0}^k (-1)^{j} a_j \\
	&\geq 0.\qedhere 
	\end{align*}
\end{proof}

\section{Consequences for learning}
\label{sec:learn}

Here we describe an application of our Fourier bounds
to learning theory;
we prove Theorem~\ref{thm:learn}.
The learning algorithm is based on standard techniques 
(see e.g.~\cite{Mansour92} or Chapter 3 in~\cite{o2014analysis}).

First, we can estimate one specific Fourier coefficient by sampling and averaging.
\begin{lemma} \label{lem:EstimateSingleFourierCoefficientBySampling}
	Let $f : \{ -1,1\}^n \to [-1,1]$ and fix a set $S \subseteq [n]$.
	Sample $X_1,\ldots,X_N \sim \{ \pm 1\}^n$ uniformly and independently and set
	$\alpha_S:= \frac{1}{N} \sum_{i=1}^Nf(X_i) \cdot \chi_S(X_i)$. Then for any $\lambda \geq 0$,
	\[
	\Pr\big[|\hat{f}(S) - \alpha_S| \geq \lambda\big] \leq 2\exp(-\lambda^2N / 2).
	\]
\end{lemma}
\begin{proof}
	Consider the random variable $Y_i := f(X_i) \cdot \chi_S(X_i)$ and note that $|Y_i| \leq 1$
	and $\E[Y_i] 
	= \hat{f}(S)$. 
	The lemma follows from the Chernoff bound (Lemma~\ref{lem:ChernovBound}).
\end{proof}


The learning algorithm operates as follows.
Its sample complexity is
$$N = \left\lceil 2\cdot 16^2 \cdot   \frac{1}{\epsilon^3} \cdot L(n,d)^2 \log \Big (2 \cdot \frac{\sum_{\ell=0}^d \binom{n}{d}}{\delta}\Big) \right\rceil$$
where
$$L(n, d):= (d+1) d^d e^{\binom{d+2}{2}} \cdot n^{\frac{d-1}{2}}.$$
The algorithm samples $X_1,\dotsc, X_N \sim \{\pm 1\}^n$ uniformly at random and independently. 
It computes $\alpha_S$ for all $S \subseteq [n]$ of size $|S| \leq d$ as in the lemma above. 
It then computes the set
$$B = \Big\{S \subseteq [n]: |\alpha_S|\geq \frac{\epsilon}{4 \cdot L(n,d)}\Big\}.$$
The output is the function 
$$g = \sum_{S \in B} \alpha_S \chi_S.$$
It remains to prove that, except with probability $\delta$, the algorithm above produces a function $g$ satisfying $\Expect{ |f(X) - g(X) |^2} \leq \epsilon$, for uniformly random $X$. 
%
%

Denote by $G$ the  event that
for every $S$ of size at most $d$
we have 
$|\hat{f}(S) - \alpha_S| \leq \frac{\epsilon^{3/2}}{16\cdot L(n,d)}$.
The union bound and Lemma~\ref{lem:EstimateSingleFourierCoefficientBySampling}
imply that $\P[G] \geq 1-\delta$.
For the rest of the proof, assume that $G$ holds.
For $S \in B$, we have $$\hat{f}(S)^2 \geq \left (\frac{\epsilon}{4 \cdot L(n,d)} - \frac{\epsilon^{3/2}}{16 \cdot L(n,d)} \right)^2 \geq \frac{\epsilon^2}{64 \cdot L(n,d)^2}.$$ 
So, by Parseval's identity, we must have that $$|B| \leq 64 \cdot  \frac{L(n,d)^2}{\epsilon^2}.$$ 
For $S \notin B$, $$|\hat{f}(S)| \leq \frac{\epsilon}{4 \cdot L(n,d)} + \frac{\epsilon^{3/2}}{16\cdot L(n,d)} \leq \frac{\epsilon}{2 \cdot L(n,d)}.$$
The last ingredient is Theorem~\ref{thml1bound}.
It implies that 
$$ \| \hat{f} \|_1 \leq   L(n,d).$$
Putting it all together,
\begin{align*}
\Expect{ |f(X) - g(X)|^2}  & = \sum_{S \subseteq [n]} (\hat{f}(S) - \hat{g}(S))^2\\ &=\sum_{S \in B} (\hat{f}(S) - \hat{g}(S))^2 + \sum_{S \notin B} \hat{f}(S)^2\\
& \leq |B| \cdot \frac{\epsilon^{3}}{16^2 \cdot L(n,d)^2}  +\frac{\epsilon}{2 \cdot L(n,d)} \cdot \sum_{S \notin B} |\hat{f}(S)|  \\
&\leq 
\epsilon.
\end{align*}

\section{Examples of bounded functions}
\label{sec:B}

In this section we provide a couple of examples showing
that our bounds are sharp for some range of parameters.

\subsection{Lower bound for $\| \hat{f}_\ell \|_1$}


Here we prove Proposition \ref{prop:const2}.
Let $\eps_S$, for $S \subseteq [n]$ of size $\ell$, be sampled uniformly and independently from $\{\pm 1\}$. Define $G(x) := \sum_{S} \eps_S \chi_S(x)$
where the sum is over $S \subseteq [n]$ of size $\ell$.
By Fact~\ref{fact:bern}, for each $x \in \{\pm 1\}^n$ we may bound $\Pr\Big[|G(x)| \ge 2 \sqrt{n \cdot {n \choose \ell}}\Big] < 2^{-n}$.
By the union bound, 
there is a choice for $\eps_S$ so that the map
$$f := \frac{G}{2\sqrt{n \cdot {n \choose \ell}}}$$ 
satisfies $\|f\|_\infty \leq 1$ and
\[
\|\hat{f}_\ell \|_1 = \frac{{n \choose \ell}}{2 \cdot \sqrt{n \cdot {n \choose \ell}}} = \frac{1}{2} \cdot \sqrt{\frac{1}{n} \cdot {n \choose \ell}}.\qedhere
\]

\subsection{Lower bound for $\|f_\ell \|_\infty$}

Here we prove Proposition \ref{prop:const1}.
Every coefficient of the Chebyshev polynomial $T_d$ is bounded by 
$\tfrac{d^d}{d!} \leq e^d$. The theorem follows from the following more general lemma.

\begin{lemma}\label{lemma:polyproject}
Given positive integers $\ell \leq d$ and a degree $d$ polynomial $T(x) = \sum_{j=0}^d c_j x^j$, define $g(x) := T((x_1 + \dots + x_n)/n)$. Then
$|g_\ell(1^n)| \geq |c_\ell| - \frac{2 (d+1)! \cdot \max_{j \ge \ell} |c_j|}{n}$.
\end{lemma}

To prove the lemma, we first show:
\begin{claim} \label{claim:hbound}
	Let $S \subseteq [n]$ be of size $\ell$, and $h_j(x) = (x_1+\dotsb+x_n)^j$. Then:
	\begin{align*}
	\hat{h}_j(S) \begin{cases}
	=0 & \text{if $j<\ell$ or $j \neq \ell \mod 2$,}\\
	= \ell! & \text{if $j = \ell$,}\\
	\leq   j! \cdot n^{\frac{j-\ell}{2}} & \text{if $j > \ell$.}
	\end{cases}
	\end{align*}
\end{claim}
\begin{proof}
	Let $X \sim \{\pm 1\}^n$ be uniformly distributed. We have
	$$ \hat{h}_j(S) = \Expect{ \chi_S(X) \cdot h_j(X)}.$$ Each $y \in [n]^j$ corresponds to the term $\prod_{i=1}^j x_{y_i}$ in the expansion of $h_j(x)$. This term contributes either $1$ or $0$ to the expectation, and it contributes $1$ exactly when every variable of $S$ has odd degree, and all other variables have even degree.  Thus, we must have $\hat{h}_j(S) = 0$ when $j< |S|$ or $j \neq |S| \mod 2$, since no term can contribute $1$ in those cases. Moreover, when $j = |S|$, we see that there are exactly $\ell!$ terms that can contribute $1$.
	
	When $j>\ell$, observe that if $\prod_{i=1}^j x_{y_i}$ contributes $1$, there must be a set $W \subseteq [j]$ of size $\ell$, such that $\prod_{i \in W} x_{y_i} = \prod_{i \in S} x_i$, and every variable of $\prod_{i \notin W} x_{y_i}$ has even degree. The number of choices for $W$ is $\binom{j}{\ell}$, and the number of ways in which $\prod_{i \in W} x_{y_i} = \prod_{i \in S} x_i$ can hold is $\ell!$. For a fixed value of $W$, the number of ways in which $\prod_{i \notin W} x_{y_i}$ can have even degrees is $\Expect{(X_1 + \dotsb + X_n)^{j-\ell}}$. Putting these observations together:
	$$ \hat{h}_j(S) \leq \binom{j}{\ell} \cdot \ell! \cdot \Expect{(X_1 + \dotsb + X_n)^{j-\ell}} \leq \binom{j}{\ell} \cdot\ell! \cdot  (j-\ell)! \cdot n^{\frac{j-\ell}{2}} = j! \cdot  n^{\frac{j-\ell}{2}};$$
	the second inequality follows from Khintchine's inequality (Fact ~\ref{fact:Khintchine}).
\end{proof}

Now we can use the claim to prove the lemma:
\begin{proof}[Proof of Lemma \ref{lemma:polyproject}]
The lemma trivially holds when $n < 4$,
so we assume $n \geq 4$. Let $h_j$ be as in Claim \ref{claim:hbound}. Note that $\hat{h}_{j}(S)\geq 0$. We can bound
\begin{align*}
g_\ell (1^n) = \| \hat{g}_\ell \|_1 &  = \sum_{S \subseteq [n], |S| = \ell} \sum_{j = 0}^{n} \hat{h}_j(S)/n^j\\
&  \geq c_\ell \binom{n}{\ell} \cdot \frac{\ell!}{n^{\ell} } -   {n \choose \ell} \cdot \sum_{j = \ell+1}^{n}   |c_j| \cdot  j! \cdot n^{-(j+\ell)/2}.
\end{align*}
To bound the first term, observe that 
\begin{align*}
c_\ell \binom{n}{\ell} \cdot \frac{\ell!}{n^{\ell} } &= c_\ell \Big(1-\frac{1}{n}\Big) \cdots \Big(1- \frac{\ell-1}{n}\Big)
\\
&\geq c_\ell \Big (1- \frac{1}{n} \cdot \sum_{j=1}^{\ell-1} j \Big)
& \text{$(*)$}\\
& = c_\ell \Big(1 - \frac{1}{n} \binom{\ell}{2} \Big);
\end{align*}
the inequality $(*)$ follows by induction from
$(1-\alpha)(1-\beta) > 1-\alpha-\beta$
for $\alpha,\beta>0$. To bound the contribution of the second term, observe:
\begin{align*}
 {n \choose \ell} \cdot \sum_{j = \ell+1}^{n}   |c_j| \cdot  j! \cdot n^{-(j+\ell)/2} &= \frac{\binom{n}{\ell}}{n^\ell} \cdot \sum_{j = \ell+2}^d  j! \cdot n^{-(j-\ell)/2} \cdot  |c_j|  \\&
\leq  \sum_{k=0}^{d-\ell-2}  j!  \cdot  \frac{1}{n} \cdot n^{-k/2}  \cdot  |c_j|  \\
&\leq \frac{1}{n} \cdot (d+1) ! \cdot \max_{k > \ell} |c_k|.
\end{align*}

Finally, since $\binom{\ell}{2} \leq (d+1)!$, we get 
$g_\ell(1^n) \geq c_\ell - \frac{2 (d+1)! \max_{k \geq \ell} |c_k|}{n}$.
\end{proof}

\section*{Acknowledgements}
We thank Mrigank Arora, Emanuel Milman, Avishay Tal, Salil Vadhan, and Kewen Wu for useful comments.

\end{document}